\documentclass[runningheads,envcountsame,a4paper]{llncs}
\usepackage[T1]{fontenc}
\usepackage[latin9]{inputenc}
\usepackage{array}
\usepackage{amsmath}
\usepackage{amssymb}
\usepackage{graphicx}

\makeatletter

\newcommand{\noun}[1]{\textsc{#1}}
\providecommand{\tabularnewline}{\\}

\usepackage{amsmath}
\usepackage{epstopdf}

\makeatother

\begin{document}
\title{Complexity of Road Coloring with Prescribed Reset Words} 
\titlerunning{Complexity of Road Coloring with Prescribed Reset Words}
\toctitle{Complexity of Road Coloring with Prescribed Reset Words}

\author{Vojt\v{e}ch Vorel\inst{1}\thanks{Supported by the Czech Science Foundation grant GA14-10799S.} \and Adam Roman\inst{2}\thanks{Supported in part by Polish MNiSW grant IP 2012 052272.}} 

\authorrunning{Vojt\v{e}ch Vorel, Adam Roman} 
\tocauthor{Vojt\v{e}ch Vorel, Adam Roman} 

\institute{Faculty of Mathematics and Physics, Charles University, Malostransk\'{e} n\'{a}m. 25, Prague, Czech Republic,\\ \email{vorel@ktiml.mff.cuni.cz}, 
\and Institute of Computer Science, Jagiellonian University,  Lojasiewicza 6, 30-348 Krakow, Poland, \\ \email{roman@ii.uj.edu.pl}}
 
\maketitle  
\setcounter{footnote}{0} 
\begin{abstract} 

By the Road Coloring Theorem (Trahtman, 2008), the edges of any given aperiodic directed multigraph with a constant out-degree can be colored such that the resulting automaton admits a reset word. There may also be a need for a particular reset word to be admitted. For certain words it is NP-complete to decide whether there is a suitable coloring. For the binary alphabet, we present a classification that separates such words from those that make the problem solvable in polynomial time. The classification differs if we consider only strongly connected multigraphs. In this restricted setting the classification remains incomplete.
\end{abstract}    

\section{Introduction}

Questions about synchronization of finite automata have been studied
since the early times of automata theory. The basic concept is very
natural: we want to find an input sequence that would get a given
machine to a unique state, no matter in which state the machine was
before. Such sequence is called a \emph{reset word}. If an automaton
has a reset word, we call it a \emph{synchronizing} \emph{automaton}. 

In the study of \emph{road coloring}, synchronizing automata are created
from directed multigraphs through edge coloring. A directed multigraph
is\emph{ }said to be\emph{ admissible}, if it is aperiodic and has
a constant out-degree. A multigraph needs to be admissible in order
to have a synchronizing coloring. Given an alphabet $I$ and an admissible
graph with out-degrees $\left|I\right|$, the following questions
arise:
\begin{enumerate}
\item Is there a coloring such that the resulting automaton has a reset
word?
\item Given a number $k\ge1$, is there a coloring such that the resulting
automaton has a reset word of length at most $k$?
\item Given a word $w\in I^{\star}$, is there a coloring such that $w$
is a reset word of the resulting automaton?
\item Given a set of words $W\subseteq I^{\star}$, is there a coloring
such that some $w\in W$ is a reset word of the resulting automaton?
\end{enumerate}
For the first question it was conjectured in 1977 by Adler, Goodwyn,
and Weiss \cite{ADL1} that the answer is always \emph{yes}. The conjecture
was known as the \emph{Road Coloring Problem} until Trahtman \cite{TRA6}
in 2008 found a proof, turning the claim into the \emph{Road Coloring
Theorem}.

The second question was initially studied in the paper \cite{ROM8conf}
presented at LATA 2012, while the yet-unpublished papers \cite{ROM8}
and \cite{VO3} give closing results: The problem is NP-complete for
any fixed $k\ge4$ and any fixed $\left|I\right|\ge2$. The instances
with $k\le3$ or $\left|I\right|=1$ can be solved by a polynomial-time
algorithm. 

The third question is the subject of the present paper. We show that
the problem becomes NP-complete even if restricted to $\left|I\right|=2$
and $w=abb$ or to $\left|I\right|=2$ and $w=aba$, which may seem
surprising. Moreover, we provide a complete classification of binary
words: The NP-completeness holds for $\left|I\right|=2$ and any $w\in\left\{ a,b\right\} ^{\star}$
that does not equal $a^{k}$, $b^{k}$, $a^{k}b$, nor $b^{k}a$ for
any $k\ge1$. On the other hand, for any $w$ that matches some of
these patterns, the restricted problem is solvable in polynomial time. 

The fourth question was raised in \cite{ROM8} and it was emphasized
that there are no results about the problem. Our results about the
third problem provide an initial step for this direction of research.

It is an easy but important remark that the Road Coloring Theorem
holds generally if and only if it holds for strongly connected graphs.
It may seem that strong connectivity can be safely assumed even if
dealing with other problems related to road coloring. Surprisingly,
we show that this does not hold for complexity issues. If P is not
equal to NP, the complexity of the third problem for strongly connected
graphs differs from the basic third problem in the case of $w=abb$.
However, for the strongly connected case we are not able to provide
a complete characterization as described above, we give only partial
results.

Due to the page limit, some proofs are omitted or shortened. The results
are presented in Sections \ref{sec:A-Complete-Classification} and
\ref{sec:A-Partial-Classification}.

\section{\label{sec:Preliminaries}Preliminaries}

\subsection{Automata and Synchronization}

For $u,w\in I^{\star}$ we say that $u$ is a \emph{prefix}, a \emph{suffix},
or a \emph{factor }of $w$ if $w=uv$, $w=vu$, or $w=vuv'$ for some
$v,v'\in I^{\star}$, respectively.

A \emph{deterministic finite automaton }is a triple $A=\left(Q,I,\delta\right)$,
where $Q$ and $I$ are finite sets and $\delta$ is an arbitrary
mapping $Q\times I\rightarrow Q$. Elements of $Q$ are called \emph{states},
$I$ is the \emph{alphabet}. The \emph{transition function} $\delta$
can be naturally extended to $Q\times I^{\star}\rightarrow Q$, still
denoted by $\delta$, slightly abusing the notation. We extend it
also by defining 
\[
\delta\!\left(S,w\right)=\left\{ \delta\!\left(s,w\right)\mid s\in S\right\} 
\]
 for each $S\subseteq Q$ and $w\in I^{\star}$. If $A=\left(Q,I,\delta\right)$
is fixed, we write $r\overset{x}{\longrightarrow}\, s$ instead of
$\delta\left(r,x\right)=s$.

For a given automaton $A=\left(Q,I,\delta\right)$, we call $w\in I^{\star}$
a \emph{reset word} if $\left|\delta\!\left(Q,w\right)\right|=1$.
If such a word exists, we call the automaton \emph{synchronizing}.
Note that each word having a reset word as a factor is also a reset
word.

\subsection{Road Coloring}

In the rest of the paper we use the term \emph{graph }for a directed
multigraph. A graph is: 
\begin{enumerate}
\item \emph{aperiodic},\emph{ }if $1$ is the only common divisor of all
the lengths of cycles,
\item \textit{admissible}\textit{\emph{,}}\textit{ }\textit{\emph{if it
is aperiodic and all its out-degrees are equal,}}
\item \textit{road colorable}\textit{\emph{,}}\textit{ }\textit{\emph{if
its edges can be labeled such that a synchronized deterministic finite
automaton arises.}}
\end{enumerate}
Naturally, we identify a coloring of edges with a transition function
$\delta$ of the resulting automaton. It is not hard to observe that
any road colorable graph is admissible. In 1977 Adler, Goodwyn, and
Weiss \cite{ADL1} conjectured that the backward implication holds
as well. Their question became known as the Road Coloring Problem
and a positive answer was given in 2008 by Trahtman \cite{TRA6}.

For any alphabet $I$ and $w\in I^{\star}$, by $\mathbb{G}_{w}^{\left|I\right|}$
we denote the set of graphs with all out-degrees equal to $\left|I\right|$
such that there exists a coloring $\delta$ with $\left|\delta\!\left(Q,w\right)\right|=1$.
In this paper we work with the following computational problem:

\renewcommand{\arraystretch}{1.6}

\begin{flushleft}
\begin{tabular}{|>{\raggedright}p{25mm}>{\raggedright}p{94mm}|}
\hline 
\multicolumn{2}{|l|}{\noun{~~SRCW }(\emph{Synchronizing road coloring with prescribed
reset words})}\tabularnewline
\textbf{~~Input:} & Alphabet $I$, graph $G=\left(Q,E\right)$ with out-degrees $\left|I\right|$,
$W\subseteq I^{\star}$\tabularnewline
\textbf{~~Output:} & Is there a $w\in W$ such that $G\in\mathbb{G}_{w}^{\left|I\right|}$?\tabularnewline[1.5mm]
\hline 
\end{tabular}\\

\par\end{flushleft}

In this paper we study the restrictions to one-element sets $W$,
which means that we consider the complexity of the sets $\mathbb{G}_{w}^{\left|I\right|}$
themselves. 

Restrictions are denoted by subscripts and superscripts: $\mathrm{SRCW}_{k,X}^{\mathcal{M}}$
denotes $\mathrm{SRCW}$ restricted to inputs with $\left|I\right|=k$,
$W=X$, and $G\in\mathcal{M}$, where $\mathcal{M}$ is a class of
graphs. By $\mathcal{SC}$ we denote the class of strongly connected
graphs. Having a graph $G=\left(Q,E\right)$ fixed, by $\mathrm{d}_{G}\!\left(s,t\right)$
we denote the length of shortest directed path from $s\in Q$ to $t\in Q$
in $G$. For each $k\ge0$ we denote 
\[
V_{k}\!\left(q\right)=\left\{ s\in Q\mid\mathrm{d}_{G}\!\left(s,q\right)=k\right\} .
\]
Having $R\subseteq Q$, let $G\!\left[R\right]$ denote the induced
subgraph of $G$ on the vertex set $R$. If a graph $G$ has constant
out-degree $\left|I\right|$, a vertex $v\in Q$ is called a \emph{sink
state }if there are $\left|I\right|$ loops on $v$. By $\mathcal{Z}$
we denote the class of graphs having a sink state. The following lemma
can be easily proved by a reduction that adds a chain of $\left|u\right|$
new states to each state of a graph:
\begin{lemma}
\label{lem:chains}Let $\left|I\right|\ge1$ and $u,w\in\left\{ a,b\right\} ^{\star}$.
Then:
\begin{enumerate}
\item If $\mathrm{SRCW}_{k,\left\{ w\right\} }$ is NP-complete, so is $\mathrm{SRCW}_{k,\left\{ uw\right\} }$.
\item If $\mathrm{SRCW}_{k,\left\{ w\right\} }^{\mathcal{Z}}$ is NP-complete,
so is $\mathrm{SRCW}_{k,\left\{ uw\right\} }^{\mathcal{Z}}$.
\end{enumerate}
\end{lemma}

\section{\label{sec:A-Complete-Classification}A Complete Classification of
Binary Words According to Complexity of $\mathrm{SRCW}_{2,\left\{ w\right\} }$}

The theorem below presents one of the main results of the present
paper. Assuming that P does not equal NP, it introduces an exact dichotomy
concerning the words over binary alphabets. Let us fix the following
partition of $\left\{ a,b\right\} ^{\star}$:
\[
\begin{array}{cc}
\begin{aligned}T_{1}= & \left\{ a^{k},b^{k}\mid k\ge0\right\} ,\\
T_{2}= & \left\{ a^{k}b,b^{k}a\mid k\ge1\right\} ,
\end{aligned}
 & \hspace{20bp}\,\begin{aligned}T_{3} & =\left\{ a^{l}b^{k},b^{l}a^{k}\mid k\ge2,l\ge1\right\} ,\\
T_{4} & =\left\{ a,b\right\} ^{\star}\backslash\left(T_{1}\cup T_{2}\cup T_{3}\right).
\end{aligned}
\end{array}
\]
For the NP-completeness reductions throughout the present paper we
use a suitable variant of the satisfiability problem. The following
can be verified using the Schaefer's dichotomy theorem \cite{SCF1}:
\begin{lemma}
\label{lem: W-SAT is NPC}It holds that W-SAT is NP-complete.
\end{lemma}
\begin{flushleft}
\begin{tabular}{|>{\raggedright}p{25mm}>{\raggedright}p{94mm}|}
\hline 
\multicolumn{2}{|l|}{\noun{~~W-SAT\label{W-SAT}}}\tabularnewline
\textbf{~~Input:} & Finite set $X$ of \emph{variables}, finite set $\Phi\subseteq X^{4}$
of \emph{clauses.}\tabularnewline
\textbf{~~Output:} & Is there an assignment $\xi:X\rightarrow\left\{ \mathbf{0},\mathbf{1}\right\} $
such that for each clause $\left(z_{1},z_{2},z_{3},z_{4}\right)\in\Phi$
it holds that:

(1) $\xi\!\left(z_{i}\right)=\mathbf{1}$ for some $i$,

(2) $\xi\!\left(z_{i}\right)=\mathbf{0}$ for some $i\in\left\{ 1,2\right\} $,

(3) $\xi\!\left(z_{i}\right)=\mathbf{0}$ for some $i\in\left\{ 3,4\right\} $?\tabularnewline[2mm]
\hline 
\end{tabular}
\par\end{flushleft}

In this section we use reductions from W-SAT to prove the NP-completeness
of $\mathrm{SRCW}_{2,\left\{ w\right\} }$ for each $w\in T_{3}$
and $w\in T_{4}$. In the case of $w\in T_{4}$ the reduction produces
only graphs having sink states. This shows that for $w\in T_{4}$
the problem $\mathrm{SRCW}_{2,\left\{ w\right\} }^{\mathcal{Z}}$
is NP-complete as well, which turns out to be very useful in Section
\ref{sec:A-Partial-Classification}, where we deal with strongly connected
graphs. For $w\in T_{3}$ we also prove NP-completeness, but we use
automata without sink states. We show that the cases with $w\in T_{1}\cup T_{2}$
are decidable in polynomial time.

\begin{flushleft}
In all the figures below we use bold solid arrows and bold dotted
arrows for the letters $a$ and $b$ respectively.
\par\end{flushleft}
\begin{theorem}
\label{thm:gen clas}Let $w\in\left\{ a,b\right\} ^{\star}$.
\begin{enumerate}
\item If $w\in T_{1}\cup T_{2}$, the problem $\mathrm{SRCW}_{2,\left\{ w\right\} }$
is solvable in polynomial time.
\item If $w\in T_{3}\cup T_{4}$, the problem $\mathrm{SRCW}_{2,\left\{ w\right\} }$
is NP-complete. Moreover, if $w\in T_{4}$, the problem $\mathrm{SRCW}_{2,\left\{ w\right\} }^{\mathcal{Z}}$
is NP-complete.
\end{enumerate}
\end{theorem}

\paragraph{Proof for $w\in T_{1}$.}

It is easy to see that $G\in\mathbb{G}_{a^{k}}$ if and only if there
is $q_{0}\in Q$ such that there is a loop on $q_{0}$ and for each
$s\in Q$ we have $\mathrm{d}_{G}\!\left(s,q_{0}\right)\le k$.\qed

\paragraph{Proof for $w\in T_{2}$.}

For a fixed $q_{0}\in Q$, we denote $Q_{1}=\left\{ s\in Q\mid s\longrightarrow q_{0}\right\} $
and 
\[
R=\left\{ s\in Q_{1}\mid H_{1}\mbox{ has a cycle reachable from }s\right\} ,
\]
where $H_{1}$ is obtained from $G\!\left[Q_{1}\right]$ by decreasing
multiplicity by $1$ for each edge ending in $q_{0}$. If $q_{0}\notin Q_{1}$,
we have $H_{1}=G\!\left[Q_{1}\right]$. Let us prove that $G\in\mathbb{G}_{a^{k}b}$
if and only if there is $q_{0}\in Q$ such that:
\begin{enumerate}
\item It holds that $\mathrm{d}_{G}\!\left(s,q_{0}\right)\le k+1$ for each
$s\in Q$.
\item For each $s\in Q$ there is a $q\in R$ such that $\mathrm{d}_{G}\!\left(s,q\right)\le k$.
\end{enumerate}
First, check the backward implication. For each $r\in R$, we color
by $b$ an edge of the form $r\longrightarrow q_{0}$ that does not
appear in $H_{1}$. Then we fix a forest of shortest paths from all
the vertices of $Q\backslash R$ into $R$. Due to the second condition
above, the branches have length at most $k$. We color by $a$ the
edges used in the forest. We have completely specified a coloring
of edges. Now, for any $s\in Q$ a prefix $a^{j}$ of $a^{k}b$ takes
us into $R$, the factor $a^{k-j}$ keeps us inside $R$, and with
the letter $b$ we end up in $q_{0}$. 

As for the forward implication, the first condition is trivial. For
the second one, take any $s\in Q$ and denote $s_{j}=\delta\!\left(s,a^{j}\right)$
for $j\ge0$. Clearly, $s_{k}\in Q_{1}$, but we show also that $s_{k}\in R$,
so we can set $q=s_{k}$ in the last condition. Indeed, whenever $s_{j}\in Q_{1}$
for $j\ge k$, we remark that $\delta\!\left(s_{j-k+1},a^{k}\right)=q_{0}$
and thus $s_{j+1}\in Q_{1}$ as well. Since $j$ can grow infinitely,
there is a cycle within $Q_{1}$ reachable from $s_{k}$.\qed

\paragraph{Proof for $w\in T_{3}$.}

Due to Lemma \ref{lem:chains}, it is enough to deal with $w=ab^{k}$
for each $k\ge2$. For a polynomial-time reduction from W-SAT, take
an instance $X=\left\{ x_{1},\dots,x_{n}\right\} $, $\Phi=\left\{ C_{1},\dots,C_{m}\right\} $,
where $C_{j}=\left(z_{j,1},z_{j,2},z_{j,3},z_{j,4}\right)$ for each
$j=1,\dots,m$. We build the graph $G_{k,\phi}=\left(Q,E\right)$
defined by Fig. \ref{fig: AphiT3}. Note that:
\begin{itemize}
\item In Fig. \ref{fig: AphiT3}, states are represented by discs. For each
$j=1,\dots,m$, the edges outgoing from $\mathrm{C}_{i}'$ and $\mathrm{C}_{i}^{''}$
represent the formula $\Phi$ by leading to the states $z_{j,1},z_{j,2},z_{j,3},z_{j,4}\in\left\{ x_{1},\dots,x_{n}\right\} \subseteq Q$. 
\item In the case of $k=2$ the state $\mathrm{V}_{i,2}$ does not exist,
so we set $x_{i}\longrightarrow\mathrm{D}_{0}$ and $\mathrm{V}_{i,1}\longrightarrow\mathrm{D}_{0}$
instead of $x_{i}\longrightarrow\mathrm{V}_{i,2}$ and $\mathrm{V}_{i,1}\longrightarrow\mathrm{V}_{i,2}$. 
\end{itemize}
We show that $G_{k,\Phi}\in\mathbb{G}_{ab^{k}}$ if and only if there
is an assignment $\xi:X\rightarrow\left\{ \mathbf{0},\mathbf{1}\right\} $
satisfying the conditions given by $\Phi$.
\begin{figure}
\begin{centering}
\includegraphics{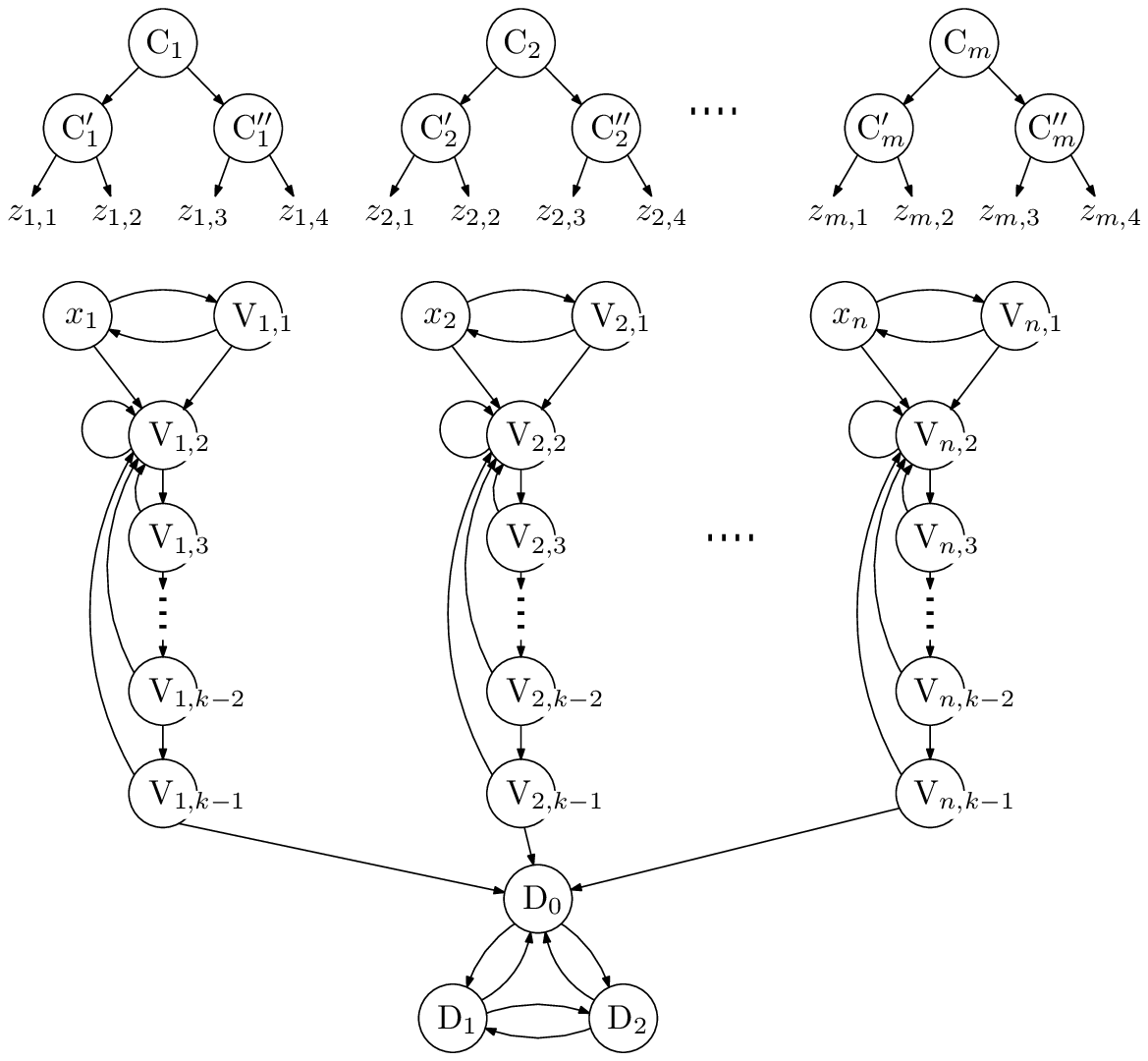}
\par\end{centering}

\caption{\label{fig: AphiT3}The graph $G_{k,\Phi}$ reducing W-SAT to $\mathrm{SRCW}_{\left|I\right|=2,W=\left\{ ab^{k}\right\} }$
for $k\ge2$}
\vspace{10mm}
\begin{minipage}[t]{0.29\columnwidth}%
\begin{center}
\includegraphics{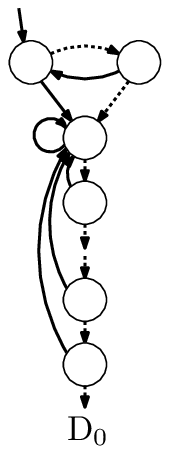}\caption{\label{fig:T3A}A coloring corresponding to $\xi\!\left(x_{i}\right)=\mathbf{0}$}

\par\end{center}%
\end{minipage}\hfill{}%
\begin{minipage}[t]{0.29\columnwidth}%
\begin{center}
\includegraphics{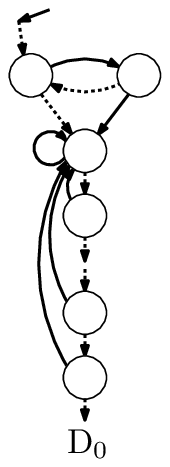}\caption{\label{fig:T3AB}A coloring corresponding to $\xi\!\left(x_{i}\right)=\mathbf{1}$}

\par\end{center}%
\end{minipage}\hfill{}%
\begin{minipage}[t]{0.33\columnwidth}%
\begin{center}
\includegraphics{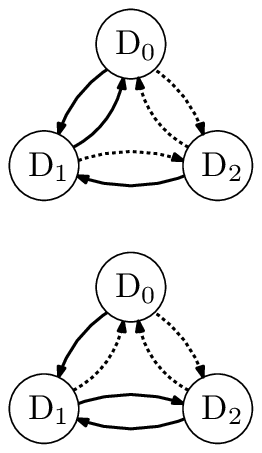}\caption{\label{fig:T3D}Colorings for $k$ even (top) and odd (bottom)}

\par\end{center}%
\end{minipage}
\end{figure}

First, let there be a coloring $\delta$ of $G_{k,\Phi}$ such that
$\left|\delta\!\left(Q,ab^{k}\right)\right|=1$. Observe that necessarily
$\delta\!\left(Q,ab^{k}\right)=\left\{ \mathrm{D}_{0}\right\} $,
while there is no loop on $\mathrm{D}_{0}$. We use this fact to observe
that whenever $x_{i}\in\delta\!\left(Q,a\right)$, the edges outgoing
from $x_{i},\mathrm{V}_{i,1},\dots,\mathrm{V}_{i,k-1}$ must be colored
according to Fig. \ref{fig:T3A}, but if $x_{i}\in\delta\!\left(Q,ab\right)$,
then they must be colored according to Fig. \ref{fig:T3AB}. Let $\xi\!\left(x_{i}\right)=\mathbf{1}$
if $x_{i}\in\delta\!\left(Q,ab\right)$ and $\xi\!\left(x_{i}\right)=\mathbf{0}$
otherwise. Choose any $j\in\left\{ 1,\dots,m\right\} $ and observe
that 
\[
\xi\!\left(\delta\!\left(\mathrm{C}_{j},ab\right)\right)=\mathbf{1},\hspace{5mm}\xi\!\left(\delta\!\left(\mathrm{C}'_{j},a\right)\right)=\mathbf{0},\hspace{5mm}\xi\!\left(\delta\!\left(\mathrm{C}''_{j},a\right)\right)=\mathbf{0},
\]
thus we can conclude that all the conditions from the definition of
W-SAT hold for the clause $C_{j}$.

On the other hand, let $\xi$ be a satisfying assignment of $\Phi$.
For each $j$ we color the edges outgoing from $\mathrm{C}_{j},\mathrm{C}'_{j},\mathrm{C}''_{j}$
such that the $ab$-path from $\mathrm{C}_{j}$ leads to the $z_{j,i}$
with $\xi\!\left(z_{j,i}\right)=\mathbf{1}$ and the $a$-paths from
$\mathrm{C}'_{j},\mathrm{C}''_{j}$ lead to the $z_{j,i'}$ and $z_{j,i''}$
with $\xi\!\left(z_{j,i'}\right)=\mathbf{0},\xi\!\left(z_{j,i''}\right)=\mathbf{0}$,
where $i'\in\left\{ 1,2\right\} ,i''\in\left\{ 3,4\right\} $. For
the edges outgoing from $x_{i},\mathrm{V}_{i,1},\dots,\mathrm{V}_{i,k-1}$
we use Fig. \ref{fig:T3A} if $\xi\!\left(x_{i}\right)=\mathbf{0}$
and Fig. \ref{fig:T3AB} if $\xi\!\left(x_{i}\right)=\mathbf{1}$.
The transitions within $\mathrm{D}_{0},\mathrm{D}_{1},\mathrm{D}_{2}$
are colored according to Fig. \ref{fig:T3D}, depending on the parity
of $k$. Observe that for each $i\in\left\{ 1,\dots,n\right\} $ we
have $x_{i}\notin\delta\!\left(Q,ab\right)$ if $\xi\!\left(x_{i}\right)=\mathbf{0}$
and $x_{i}\notin\delta\!\left(Q,a\right)$ if $\xi\!\left(x_{i}\right)=\mathbf{1}$.
Using this fact we check that $\delta\!\left(Q,w\right)=\left\{ \mathrm{D}_{0}\right\} $.
\qed

\paragraph{Proof for $w\in T_{4}$.}

Any $w\in T_{4}$ can be written as $w=va^{j}b^{k}a^{l}$ or $w=vb^{j}a^{k}b^{l}$
for $j,k,l\ge1$. Due to Lemma \ref{lem:chains} it is enough to deal
with $w=ab^{k}a^{l}$ for each $k,l\ge1$. Take an instance of W-SAT
as above and construct the graph $G_{w,\Phi}=\left(Q,E\right)$ defined
by Fig. \ref{fig:AphiT4}. Note that:
\begin{itemize}
\item In the case of $l=1$, the state $\mathrm{Z}_{i,1}$ does not exist,
so we set $\mathrm{W}_{i}'\longrightarrow\mathrm{D}_{0}$ and $\mathrm{V}_{i,k-1}\longrightarrow\mathrm{D}_{0}$
instead of $\mathrm{W}_{i}'\longrightarrow\mathrm{Z}_{i,1}$ and $\mathrm{V}_{i,k-1}\longrightarrow\mathrm{Z}_{i,1}$.
\item In the case of $k=1$, the state $V_{i,1}$ does not exist, so we
set $x_{i}\longrightarrow\mathrm{Z}_{i,1}$ (or $x_{i}\longrightarrow\mathrm{D}_{0}$
if $l=1$) and $x_{i}\longrightarrow\mathrm{W}_{i}$ instead of $x_{i}\rightrightarrows\mathrm{V}_{i,1}$.
\begin{figure}
\begin{centering}
\includegraphics{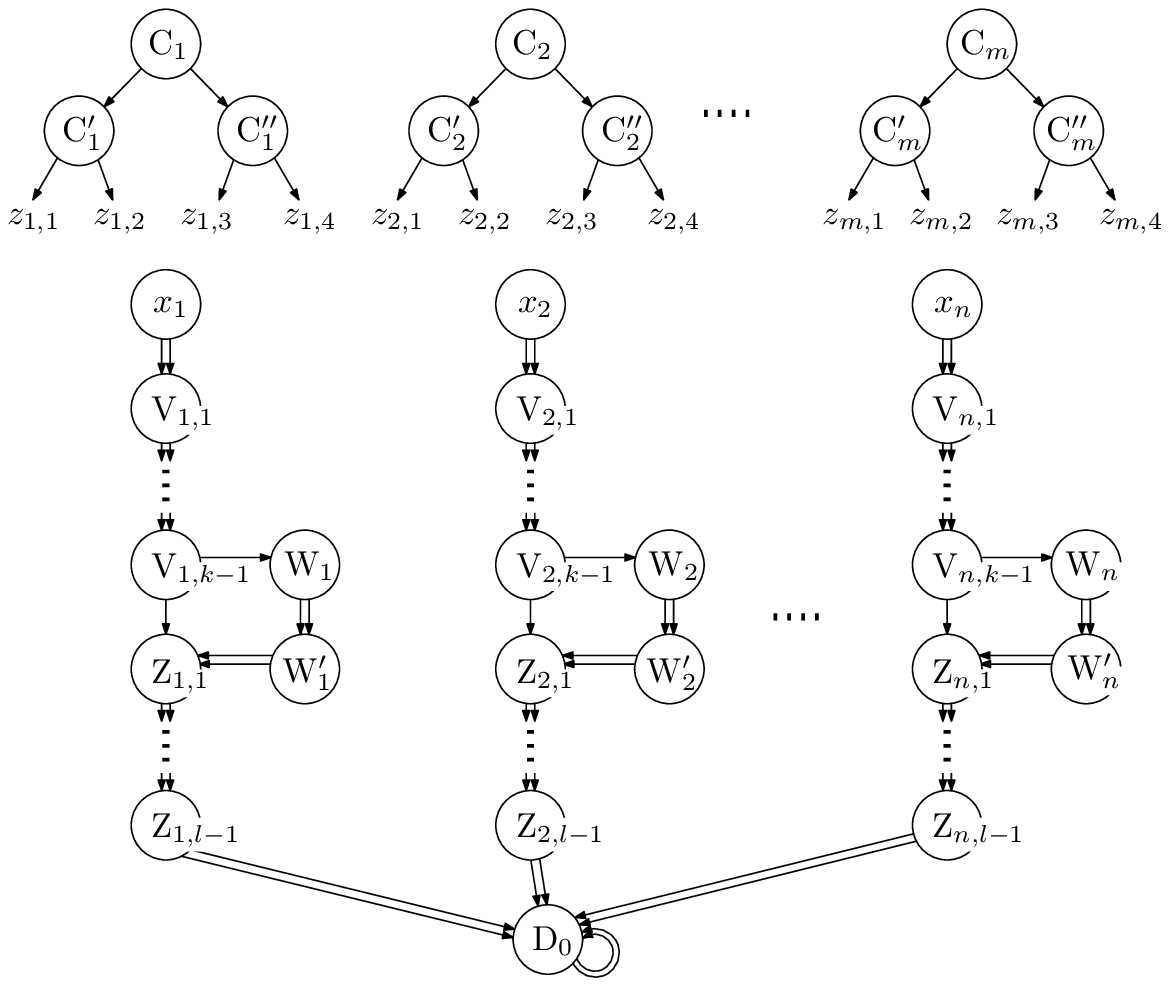}
\par\end{centering}

\caption{\label{fig:AphiT4}The graph $G_{w,\Phi}$ reducing W-SAT to $\mathrm{SRCW}_{\left|I\right|=2,W=\left\{ ab^{k}a^{l}\right\} }^{\mathcal{Z}}$
for $k,l\ge1$}
\end{figure}

\end{itemize}
Let there be a coloring $\delta$ of $G_{w,\Phi}$ such that $\left|\delta\!\left(Q,w\right)\right|=1$.
Observe that $\delta\!\left(Q,w\right)=\left\{ \mathrm{D}_{0}\right\} $.
Next, observe that whenever $x_{i}\in\delta\!\left(Q,a\right)$, then
$\mathrm{V}_{i,k-1}\overset{b}{\longrightarrow}\mathrm{Z}_{i,1}$,
but if $x_{i}\in\delta\!\left(Q,ab\right)$, then $\mathrm{V}_{i,k-1}\overset{a}{\longrightarrow}\mathrm{Z}_{i,1}$.
Let $\xi\!\left(x_{i}\right)=\mathbf{1}$ if $x_{i}\in\delta\!\left(Q,ab\right)$
and $\xi\!\left(x_{i}\right)=\mathbf{0}$ otherwise. We choose any
$j\in\left\{ 1,\dots,m\right\} $ and conclude exactly as we did in
the case of $T_{3}$.

On the other hand, let $\xi$ be a satisfying assignment of $\Phi$.
For each $j$, we color the edges outgoing from $\mathrm{C}_{j},\mathrm{C}'_{j},\mathrm{C}''_{j}$
as we did in the case of $T_{3}$. For each $i$, we put $\mathrm{V}_{i,k-1}\overset{a}{\longrightarrow}\mathrm{Z}_{i,1},\mathrm{V}_{i,k-1}\overset{b}{\longrightarrow}\mathrm{W}_{i}$
if $\xi\!\left(x_{i}\right)=\mathbf{1}$ and the reversed variant
if $\xi\!\left(x_{i}\right)=\mathbf{0}$.\qed

\section{\label{sec:A-Partial-Classification}A Partial Classification of
Binary Words According to Complexity of $\mathrm{SRCW}_{2,\left\{ w\right\} }^{\mathcal{SC}}$}

Clearly, for any $w\in T_{1}\cup T_{2}$ we have $\mathrm{SRCW}_{2,\left\{ w\right\} }^{\mathcal{SC}}\in\mathrm{P}$.
In Section \ref{sub:abb sc in p} we show that 
\[
\mathrm{SRCW}_{2,\left\{ abb\right\} }^{\mathcal{SC}}\in\mathrm{P},
\]
which is a surprising result because the general $\mathrm{SRCW}_{2,\left\{ w\right\} }$
is NP-complete for any $w\in T_{3}$, including $w=abb$. We are not
aware of any other words that witness this difference between $\mathrm{SRCW}^{\mathcal{SC}}$
and $\mathrm{SRCW}$. 

In Section \ref{sub:sink devs} we introduce a general method using
\emph{sink devices }that allows us to prove the NP-completeness of
$\mathrm{SRCW}_{2,\left\{ w\right\} }^{\mathcal{SC}}$ for infinitely
many words $w\in T_{4}$, including any $w\in T_{4}$ with the first
and last letter being the same. However, we are not able to apply
the method to each $w\in T_{4}$.

\subsection{\label{sub:abb sc in p}A Polynomial-Time Case}

A graph $G=\left(Q,E\right)$ is said to be \emph{$k$-lifting} if
there exists $q_{0}\in Q$ such that for each $s\in Q$ there is an
edge leading from $s$ into $V_{k}\!\left(q_{0}\right)$. Instead
of $2$-lifting we just say \emph{lifting}.
\begin{lemma}
\label{lem:if raked then abb }If $G$ is a $k$-lifting graph, then
$G\in\mathbb{G}_{ab^{k}}$.
\end{lemma}

\begin{lemma}
\label{lem:no inc b}If $G$ is strongly connected, $G$ is not lifting,
and $G\in\mathbb{G}_{abb}$ via $\delta$ and $q_{0}$, then $\delta$
has no $b$-transition ending in $V_{2}\!\left(q_{0}\right)\cup V_{3}\!\left(q_{0}\right)$.
Moreover, $V_{3}\!\left(q_{0}\right)=\emptyset$.\end{lemma}
\begin{proof}
First, suppose for a contradiction that some $s\in V_{2}\!\left(q_{0}\right)\cup V_{3}\!\left(q_{0}\right)$
has an incoming $b$-transition. Together with its outgoing $b$-transition
we have
\[
r\overset{b}{\longrightarrow}s\overset{b}{\longrightarrow}t,
\]
where $s\neq q_{0}$ and $t\neq q_{0}$. Due to the strong connectivity
there is a shortest path $P$ from $q_{0}$ to $r$ (possibly of length
$0$ if $r=q_{0}$). The path $P$ is made of $b$-transitions. Indeed,
if there were some $a$-transitions, let $r'\overset{a}{\longrightarrow}r''$
be the last one. The $abb$-path outgoing from $r'$ ends in $\delta\!\left(r'',bb\right)$,
which either lies on $P$ or in $\left\{ s,t\right\} $, so it is
different from $q_{0}$ and we get a contradiction. 

It follows that $\delta\!\left(q_{0},b\right)\neq q_{0}$ and $\delta\!\left(q_{0},bb\right)\neq q_{0}$,
so there cannot be any $a$-transition incoming to $q_{0}$. Hence
for any $s\in V_{1}\!\left(q_{0}\right)$ there is a transition $s\overset{b}{\longrightarrow}q_{0}$
and thus there is no $a$-transition ending in $V_{1}\!\left(q_{0}\right)$.
Because there is also no $a$-transition ending in $V_{3}\!\left(q_{0}\right)$,
all the $a$-transitions end in $V_{2}\!\left(q_{0}\right)$ and thus
$G$ is lifting, which is a contradiction.

Second, we show that $V_{3}\!\left(q_{0}\right)$ is empty. Suppose
that $s\in V_{3}\!\left(q_{0}\right)$. No $a$-transition comes to
$s$ since there is no path of length $2$ from $s$ to $q_{0}$.
Thus, $s$ has no incoming transition, which contradicts the strong
connectivity.\qed

\end{proof}
\begin{theorem}
\label{thm: abb SC in P}$\mathrm{SRCW}_{2,\left\{ abb\right\} }^{\mathcal{SC}}$
is decidable in polynomial time.\end{theorem}
\begin{proof}
As the input we have a strongly connected $G=\left(Q,E\right)$. Suppose
that $q_{0}$ is fixed (we can just try each $q_{0}\in Q$) and so
we should decide if there is some $\delta$ with $\delta\!\left(Q,abb\right)=\left\{ q_{0}\right\} $.
First we do some preprocessing:
\begin{itemize}
\item If $G$ is lifting, according to Lemma \ref{lem:if raked then abb }
we accept.
\item If $V_{3}\!\left(q_{0}\right)\neq\emptyset$, according to Lemma \ref{lem:no inc b}
we reject.
\item If there is a loop on $q_{0}$, we accept, since due to $V_{3}\!\left(q_{0}\right)=\emptyset$
we have $G\in\mathbb{G}_{bb}$.
\end{itemize}

If we are still not done, we try to find some labeling $\delta$,
assuming that none of the three conditions above holds. We deduce
two necessary properties of $\delta$. First, Lemma \ref{lem:no inc b}
says that we can safely label all the transitions ending in $V_{2}\!\left(q_{0}\right)$
by $a$. Second, we have $q_{0}\in\delta\!\left(Q,a\right)$. Indeed,
otherwise all the transitions incoming to $q_{0}$ are labeled by
$b$, and there cannot be any $a$-transition ending in $V_{1}\!\left(q_{0}\right)$
because we know that the $b$-transition outgoing from $q_{0}$ is
not a loop. Thus $G$ is lifting, which is a contradiction.

Let the sets $B_{1},\dots,B_{\beta}$ denote the connected components
(not necessarily strongly connected) of $G\!\left[V_{1}\!\left(q_{0}\right)\right]$.
Note that maximum out-degree in $G\!\left[V_{1}\!\left(q_{0}\right)\right]$
is $1$. Let $e=\left(r,s\right),e'=\left(s,t\right)$ be consecutive
edges with $s,t\in V_{1}\!\left(q_{0}\right)$ and $r\in Q$. Then
the labeling $\delta$ has to satisfy
\[
e\mbox{ is labeled by }a\mbox{ }\Leftrightarrow\mbox{ }e'\mbox{ is labeled by }b.
\]

Indeed:
\begin{itemize}
\item The left-to-right implication follows easily from the fact that there
is no loop on $q_{0}$. 
\item As for the other one, suppose for a contradiction that both $e',e$
are labeled by $b$. We can always find a path $P$ (possibly trivial)
that starts outside $V_{1}\!\left(q_{0}\right)$ and ends in $r$.
Let $\overline{r}$ be the last vertex on $P$ that lies in $\delta\!\left(Q,a\right)$.
Such vertex exists because we have $V_{2}\!\left(q_{0}\right)\cup\left\{ q_{0}\right\} \subseteq\delta\!\left(Q,a\right)$
and $V_{3}\!\left(q_{0}\right)=\emptyset$. Now we can deduce that
$\delta\!\left(\overline{r},bb\right)\neq q_{0}$, which is a contradiction.
\end{itemize}

It follows that for each $B_{i}$ there are at most two possible colorings
of its inner edges (fix \emph{variant $\mathbf{0}$ }and\emph{ variant
$\mathbf{1}$} arbitrarily). Moreover, a labeling of any edge incoming
to $B_{i}$ enforces a particular variant for whole $B_{i}$.

Let the set $A$ contain the vertices $s\in V_{2}\!\left(q_{0}\right)\cup\left\{ q_{0}\right\} $
whose outgoing transitions lead both into $V_{1}\!\left(q_{0}\right)$.
Edges that start in vertices of $\left(V_{2}\!\left(q_{0}\right)\cup\left\{ q_{0}\right\} \right)\backslash A$
have only one possible way of coloring due to Lemma \ref{lem:no inc b},
while for each vertex of $A$ there are two possibilities. Now any
possible coloring can be described by $\left|A\right|+\beta$ Boolean
propositions:
\begin{eqnarray*}
\mathbf{x}_{s} & \equiv & e_{s}\mbox{ is labeled by }a\\
\mathbf{y}_{B} & \equiv & B\mbox{ is labeled according to variant }\mathbf{1}
\end{eqnarray*}
for each $s\in A$ and $B\in\left\{ B_{1},\dots,B_{\beta}\right\} $,
where $e_{s}$ is a particular edge outgoing from $s$. Moreover,
the claim $\delta\!\left(Q,abb\right)=\left\{ q_{0}\right\} $ can
be equivalently formulated as a conjunction of implications of the
form $\mathbf{x}_{s}\rightarrow\mathbf{y}_{B}$, so we reduce the
problem to 2-SAT.\qed

\end{proof}

\subsection{\label{sub:sink devs}NP-Complete Cases}

We introduce a method based on \emph{sink devices }to prove the NP-completeness
for a wide class of words even under the restriction to strongly connected
graphs. 

In the proofs below we use the notion of a \emph{partial finite automaton}
(\emph{PFA}), which can be defined as a triple $P=\left(Q,I,\delta\right)$,
where $Q$ is a finite set of states, $I$ is a finite alphabet, and
$\delta$ is a partial function $Q\times I\rightarrow Q$ which can
be naturally extended to $Q\times I^{\star}\rightarrow Q$. Again,
we write $r\overset{x}{\longrightarrow}\, s$ instead of $\delta\left(r,x\right)=s$.
We say that a PFA is \emph{incomplete }if there is some undefined
value of $\delta$. A \emph{sink state }in a PFA has a defined loop
for each letter.
\begin{definition}
\label{def:sink device}Let $w\in\left\{ a,b\right\} ^{\star}$. We
say that a PFA $B=\left(Q,\left\{ a,b\right\} ,\delta\right)$ is
a \emph{sink device }for $w$, if there exists $q_{0}\in Q$ such
that:
\begin{enumerate}
\item $\delta\!\left(q_{0},u\right)=q_{0}$ for each prefix $u$ of $w$,
\item $\delta\!\left(s,w\right)=q_{0}$ for each $s\in Q$.
\end{enumerate}
\end{definition}
Note that the trivial automaton consisting of a single sink state
is a sink device for any $w\in\left\{ a,b\right\} ^{\star}$. However,
we are interested in strongly connected sink devices that are incomplete.
In Lemma \ref{lem:pouziti dev} we show how to prove the NP-completeness
using a non-specific sink device in the general case of $w\in T_{4}$
and after that we construct explicit sink devices for a wide class
of words from $T_{4}$.
\begin{lemma}
\label{lem:pouziti dev}Let $w\in T_{4}$ and assume that there exists
a strongly connected incomplete sink device $B$ for $w$. Then $\mathrm{SRCW}_{2,\left\{ w\right\} }^{\mathcal{SC}}$
is NP-complete.\end{lemma}
\begin{proof}
We assume that $w$ starts by $a$ and write $w=a^{\alpha}b^{\beta}au$
for $\alpha,\beta\ge1$ and $u\in\left\{ a,b\right\} ^{\star}$. Denote
$B=\left(Q_{B},\left\{ a,b\right\} ,\delta_{B}\right)$. For a reduction
from W-SAT, take an instance $X,\Phi$\textbf{ }with the notation
used before, assuming that each $x\in X$ occurs in $\Phi$. We construct
a graph $\overline{G}_{w,\Phi}=\left(\overline{Q},\overline{E}\right)$
as follows. Let $q_{1}\in Q_{B}$ have an undefined outgoing transition,
and let $B'$ be an automaton obtained from $B$ by arbitrarily defining
all the undefined transitions except for one transition outgoing from
$q_{1}$. Let $G_{B'}$ be the underlying graph of $B'$. By Theorem
\ref{thm:gen clas}, $\mathrm{SRCW}_{2,\left\{ w\right\} }^{\mathcal{Z}}$
is NP-complete, so it admits a reduction from W-SAT. Let $G_{w,\Phi}=\left(Q,E\right)$
be the graph obtained from such reduction, removing the loop on the
sink state $q'_{0}\in Q$. Let $s_{1},\dots,s_{\left|Q\right|-1}$
be an enumeration of all the states of $G_{w,\Phi}$ different from
$q'_{0}$. Then we define $\overline{G}_{w,\Phi}$ as shown in Fig.
\ref{fig:scAll}. We merge the state $q'_{0}\in Q$ with the state
$q_{0}\in Q_{B}$, which is fixed by the definition of a sink device.
\textbf{}
\begin{figure}
\begin{centering}
\includegraphics{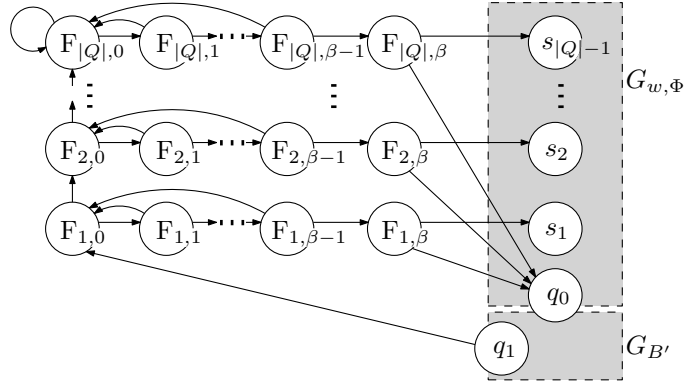}
\par\end{centering}

\textbf{\caption{\label{fig:scAll}The graph $\overline{G}_{w,\Phi}$}
}

\end{figure}

First, let there be a coloring $\overline{\delta}$ of $\overline{G}_{w,\Phi}$
such that $\left|\overline{\delta}\!\left(\overline{Q},w\right)\right|=1$.
It follows easily that $\overline{\delta}$, restricted to $Q$, encodes
a coloring $\delta$ of $G_{w,\Phi}$ such that $\left|\delta\!\left(Q,w\right)\right|=1$.
The choice of $G_{w,\Phi}$ guarantees that there is a satisfying
assignment $\xi$ for $\Phi$.

On the other hand, let $\xi$ be a satisfying assignment of $\Phi$.
By the choice of $G_{w,\Phi}$, there is a coloring $\delta$ of $G_{w,\Phi}$
such that $\left|\delta\!\left(Q,w\right)\right|=1$. We use the following
coloring of $\overline{G}_{w,\Phi}$: The edges outgoing from $s_{1},\dots,s_{\left|Q\right|-1}$
are colored according to $\delta$. The edges within $G_{B'}$ are
colored according to $B'$. The edge $q_{1}\longrightarrow\mathrm{F}_{1,0}$
is colored by $b$. All the other edges incoming to the states $\mathrm{F}_{1,0},\dots,\mathrm{F}_{\left|Q\right|,0}$,
together with the edges of the form $\mathrm{F}_{i,\beta}\longrightarrow q_{0}$,
are colored by $a$, while the remaining ones are colored by $b$.
\qed

\end{proof}
For any $w\in\left\{ a,b\right\} ^{\star}$ we construct a strongly
connected sink device $\mathbf{D}\!\left(w\right)=\left(Q_{w},\left\{ a,b\right\} ,\delta_{w}\right)$.
However, for some words $w\in T_{4}$ (e.g. for $w=abab$) the device
$\mathbf{D}\!\left(w\right)$ is not incomplete and thus is not suitable
for the reduction above. Take any $w\in\left\{ a,b\right\} ^{\star}$
and let $\mathfrak{C}_{w}^{\mathrm{P}},\mathfrak{C}_{w}^{\mathrm{S}},\mathfrak{C}_{w}^{\mathrm{F}}$
be the sets of all prefixes, suffixes and factors of $w$ respectively,
including the empty word $\epsilon$. Let 
\begin{eqnarray*}
Q_{w} & = & \left\{ \left[u\right]\mid u\in\mathfrak{C}_{w}^{\mathrm{F}},\, v\notin\mathfrak{C}_{w}^{\mathrm{S}}\mbox{ for each nonempty prefix }v\mbox{ of }u\right\} ,
\end{eqnarray*}
while the partial transition function $\delta_{w}$ consists of the
following transitions: 
\begin{enumerate}
\item \label{enu: def of D(w)}$\left[u\right]\overset{x}{\longrightarrow}\left[ux\right]$
whenever $\left[u\right],\left[ux\right]\in Q_{w}$,
\item $\left[u\right]\overset{x}{\longrightarrow}\left[\epsilon\right]$
whenever $ux\in\mathfrak{C}_{w}^{\mathrm{S}}$,
\item $\left[u\right]\overset{x}{\longrightarrow}\left[\epsilon\right]$
whenever $\left[ux\right]\notin Q_{w}$, $ux\notin\mathfrak{C}_{w}^{\mathrm{S}}$,
and $vx\in\mathfrak{C}_{w}^{\mathrm{P}}$ for a suffix $v$ of $u$.\end{enumerate}
\begin{lemma}
\label{lem:D je sink}For any $w\in\left\{ a,b\right\} ^{\star}$,
$\mathbf{D}\!\left(w\right)$ is a strongly connected sink device.
\end{lemma}

\begin{lemma}
\label{lem:conds for incompl}Suppose that $w\in\left\{ a,b\right\} ^{\star}$
starts by $x$, where $\left\{ x,y\right\} =\left\{ a,b\right\} $.
If there is $u\in\left\{ a,b\right\} ^{\star}$ satisfying all the
following conditions, then $\mathbf{D}\!\left(w\right)$ is incomplete:
\begin{enumerate}
\item $\left[u\right]\in Q_{w}$,
\item $uy\notin\mathfrak{C}_{w}^{\mathrm{F}}$,
\item for each nonempty suffix $v$ of $uy$, $v\notin\mathfrak{C}_{w}^{\mathrm{P}}$.
\end{enumerate}
\end{lemma}

\begin{theorem}
If a word $w\in T_{4}$ satisfies some of the following conditions,
then $\mathrm{SRCW}_{2,\left\{ w\right\} }^{\mathcal{SC}}$ is NP-complete:
\begin{enumerate}
\item $w$ is of the form $w=x\overline{w}x$ for $\overline{w}\in\left\{ a,b\right\} ^{\star},x\in\left\{ a,b\right\} $,
\item $w$ is of the form $w=x\overline{w}y$ for $\overline{w}\in\left\{ a,b\right\} ^{\star},x,y\in\left\{ a,b\right\} ,x\neq y$,\\
and $x^{k}y^{l}x\in\mathfrak{C}_{w}^{\mathrm{F}},\, x^{k+1}\notin\mathfrak{C}_{w}^{\mathrm{F}},\, y^{l+1}\notin\mathfrak{C}_{w}^{\mathrm{F}}$
for some $k,l\ge1$.
\end{enumerate}
\end{theorem}
\begin{proof}
Due to Lemmas \ref{lem:pouziti dev} and \ref{lem:D je sink}, it
is enough to show that $\mathbf{D}\!\left(w\right)$ is incomplete.
Let $m\ge1$ be the largest integer such that $y^{m}$ is a factor
of $w$. It is straightforward to check that $u=y^{m}$ (in the first
case) or $u=x^{k}y^{l}$ (in the second case) satisfies the three
conditions from Lemma \ref{lem:conds for incompl}.
\end{proof}

\section{Conclusion and Future Work}

We have completely characterized the binary words $w$ that make the
computation of road coloring NP-complete if some of them is required
to be the reset word for a coloring of a given graph. Except for $w=a^{k}b$
and $w=a^{k}$ with $k\ge1$, each $w\in\left\{ a,b\right\} ^{\star}$
has this property. We have proved that if we require strong connectivity,
the case $w=abb$ becomes solvable in polynomial time. For any $w$
such that the first letter equals to the last one and both $a,b$
occur in $w$, we have proved that the NP-completeness holds even
under this requirement. The main goals of the future research are:
\begin{itemize}
\item Complete the classification of binary words in the strongly connected
case.
\item Give the classifications of words over non-binary alphabets.
\item Study SRCW restricted to non-singleton sets of words.
\end{itemize}
\bibliographystyle{splncs03}
\bibliography{ruco}

\end{document}